\documentclass[acmlarge]{acmart}
\AtBeginDocument{%
  \providecommand\BibTeX{{%
    \normalfont B\kern-0.5em{\scshape i\kern-0.25em b}\kern-0.8em\TeX}}}

\usepackage{amsmath}
\usepackage{amsthm}
\newtheorem{remark}{remark}
\usepackage{algorithm}
\usepackage{algpseudocode}
\usepackage{textcomp}
\usepackage{multirow}

\begin{document}

\title{On a Geometric Interpretation Of the Subset Sum Problem}
\subtitle{ A Novel FPTAS for SSP }
%
\author{Marius Costandin}
\email{costandinmarius@gmail.com}
\begin{abstract}
For $S \in \mathbb{N}^n$ and $T \in \mathbb{N}$, the Subset Sum Problem (SSP) $\exists^? x \in \{0,1\}^n $ such that $S^T\cdot x  = T$ can be interpreted as the problem of deciding whether the intersection of the positive unit hypercube $Q_n = [0,1]^n$ with the hyperplane $S^T\cdot \left(x - \frac{S}{\|S\|^2 }\cdot T \right) = 0$ contains at least a vertex. In this paper, we give an algorithm of complexity $\mathcal{O}\left( \frac{1}{\epsilon}\cdot n^b \right)$, for some absolute constant $b$, which either proves that there are no vertices in a slab of thickness $\epsilon$ either finds a vertex in the slab of thickness $4\cdot \epsilon$. It is shown that any vertex $P$ in a slab of thickness $\epsilon$ meets $\left| \frac{S^T\cdot P}{T} - 1 \right| \leq \epsilon$, therefore making the proposed algorithm a FPTAS for the SSP. The results are  then applied to the study of the so called Simultaneous Subset-Sum Problem (SSSP). 
\end{abstract}

%

\keywords{Non-Convex Optimization, NP-Complete, Subset-Sum}


\maketitle

\section{Introduction}

This short paper is concerned with the study of the classical Subset Sum Problem. That is, for $m,n \in \mathbb{N}$, $S \in \mathbb{N}^n$ and $T \in \mathbb{N}$ one defines the problem:
\begin{align}\label{E1a}
\exists^? x \in \{0,1\}^n \hspace{0.5cm} \text{with} \hspace{0.5cm} S^T\cdot x = T. 
\end{align} where each element of $S$, $s_k \leq 2^m$ for $k \in \{1, \hdots, n\}$. 

As already stated in the paper abstract, the SSP can be interpreted as the problem of deciding whether the intersection of the positive unit hypercube $Q_n = [0,1]^n$ with the hyperplane $S^T\cdot \left(x - \frac{S}{\|S\|^2 }\cdot T \right) = 0$ contains at least a vertex. It can be shown that the SSP can be reduced to the Subset Partition Problem (SPP), i.e. a particular instance of the SSP where the target $T = \frac{S^T\cdot 1_{n \times 1}}{2}$. As such (\ref{E1a}) becomes:
 
\begin{align}\label{E2a}
\exists^? x \in \{0,1\}^n \hspace{0.5cm} \text{with} \hspace{0.5cm} S^T\cdot \left(x - \frac{1}{2} \cdot 1_{n\times 1}\right) = 0
\end{align} 

On the other hand, looking at (\ref{E1a}) form the perspective the Dynamical Programming, one can see that that the difficulty lies in the size of $T = \frac{S^T\cdot 1_{n\times 1}}{2}$ which basically is determined by the size of the elements of $S$. 

As such, one can conclude that the difficulty of solving (\ref{E2a}) is due to the fact that the hyperplane normal vector $S$ can basically have many orientations.    The main idea of this work is to approximate the normal vector $S$ with the vector $U$ which has elements written on less bits, hence one can use the Dynamical Programming Algorithm to determine whether the new hyperplane $U^T \cdot \left(x - \frac{1}{2} \cdot 1_{n\times 1}\right) = 0 $ contains a vertex of the hypercube.  

Although it is obvious that the vector $U$ cannot have as many directions as $S$ would have, we show that at most $\mathcal{O}(n)$ small parallel translation of the hyperplane $U^T \cdot \left(x - \frac{1}{2} \cdot 1_{n\times 1}\right) = 0 $ contain the vertices that would belong to the hyperplane $S^T \cdot \left(x - \frac{1}{2} \cdot 1_{n\times 1}\right) = 0 $. It is therefore possible to use $U$ to state some results of a relaxation of the geometric interpretation of the subset sum: instead of asking if the unit hypercube has vertices on the infinitely thin hyperplane slicing it, we ask whether the hypercube has vertices on an "$\epsilon$-thick slice" that we shall call "slab" as defined below:  

\begin{definition}
Given the hyperplane $S^T\cdot (x - C) = 0$, we call a slab of the unit hypercube $\mathcal{Q}_n = [0,1]^n$ of thickness $\delta > 0$ the set 

\begin{align}
&\mathcal{S}(S,C, \delta) = \left\{x \in \mathcal{Q}_n\biggr| \frac{S}{\|S\|}^T \cdot \left( x -  C \right) \leq \frac{\delta}{2}, \frac{S}{\|S\|}^T\cdot \left( x - C\right) \geq -\frac{\delta}{2} \right\} 
\end{align} 

\end{definition}

The main results of the section MAIN RESULTS are summarized in the following theorem:
\begin{theorem}
For $d^{\star}$ given by (\ref{E11}), there exists an algorithm of complexity $\mathcal{O}\left( N\cdot n^{\frac{5}{2}}\right)$ which upon completion fulfils exactly one of the following alternatives:
\begin{enumerate}
\item proves that the slab $\mathcal{S}\left( S, \frac{1}{2}\cdot 1_{n \times 1},2\cdot d^{\star}   \approx  \frac{n}{N} \right)$ contains no vertex of the hypercube $Q_n = [0,1]^n$. 
\item finds a vertex of the hypercube  $Q_n = [0,1]^n$ in the slab \\$\mathcal{S}\left( S, \frac{1}{2}\cdot 1_{n \times 1}, 8\cdot d^{\star}   \approx  \frac{4\cdot n}{N} \right)$. 
\end{enumerate}
\end{theorem} where $N$ is the largest element in $U$, $\frac{d^{\star}}{\frac{n}{N}} \sim 1 $ as $\frac{n}{N^2} \to 0$. We take $N = n^c$ for some fixed $2 \leq c \in \mathbb{N}$ arbitrarily chosen by user.  

Since for  $P \in \{0,1\}^n \cap \mathcal{S}\left( S, \frac{1}{2}\cdot 1_{n \times 1}, \delta^{\star} \right)$ one has the existence of $\delta \in [-\delta^{\star}, \delta^{\star}]$ such that
\begin{align}
\frac{S^T}{\|S\|} \cdot \left( P - \left( \frac{1}{2} \cdot 1_{n \times 1} + \frac{S}{\|S\|}\cdot \delta \right)\right) = 0
\end{align} i.e. 
\begin{align}
S^T\cdot P = \frac{S^T\cdot 1_{n \times 1}}{2} + \delta \cdot \|S\| \hspace{0.2cm} \Rightarrow \hspace{0.2cm} \left| \frac{S^T\cdot P}{\frac{S^T\cdot 1_{n\times 1}}{2}} - 1 \right| \leq 2\cdot \delta^{\star}
\end{align} follows that the algorithm given in the above Theorem is a FPTAS for the SPP.

Work in the literature for the Subset Sum problem can be found in \cite{DP1}, \cite{DP2}, \cite{DP3}, where the authors also use dynamical programming, \cite{DP4} where the authors present an implementation on GPU, as well as in \cite{ML1a} where a short FPTAS is given based on a list representation of the SSP. A comprehensive review of the literature can be found in \cite{KP} where the authors study several knapsack problems (of which the SSP is a particular case). 

An interesting approach to NP-complete problems, as SSP is, are quantum algorithms. In \cite{QC1}, \cite{QC2} the reader can find several approaches to this emerging computation technology.   

In the next section APPLICATION TO SIMULTANEOUS SUBSET SUM PROBLEM we study the problem of deciding the existence of a vertex of the positive hypercube which is simultaneously contained in more than one and less than $n$ slabs. For this we give similar results for thick spherical shells of large radius spheres by approximating them locally with slabs. Define  
\begin{align}
L_0(x) = \sum_{i = 1}^p \left( \|x - C_i\|^2 - R_i^2 \right)^2 \nonumber 
\end{align} and we have the main theorem of this section for $p$ simultaneous subset sum problems. 

\begin{theorem} Assume that $p < n$ is a power of two. If $\left\| \frac{\sum_{i=1}^p C_i}{p} - C \right\| $ is large enough, then for a given $\delta > 0$ there exists an algorithm of quasi-polynomial complexity which finds  $ x \in \{0,1\}^n$ such that $L_0(x) \in \mathcal{O}( \delta)$ if \textit{exactly one} exists. Failing to return a solution therefore means that either no solution exists either there are at least two. 
\end{theorem}

\section{Main Results}

For $m,n \in \mathbb{N}$, let $s_1, ..., s_n \in \mathbb{N}$ with $s_k \leq 2^m$. W.l.o.g. consider the partition problem (a special case of the Subset Sum problem):
\begin{align}\label{E1}
\exists^? x \in \{0,1\}^n \hspace{0.5cm} \text{with} \hspace{0.5cm} S^T\cdot \left(x - \frac{1}{2} \cdot 1_{n\times 1}\right) = 0
\end{align} where $S = \begin{bmatrix} s_1 &... &s_n\end{bmatrix}^T$.

This can be interpreted, as the problem of determining whether the hyperplane $S^T\cdot \left(x - \frac{1}{2} \cdot 1_{n\times 1}\right) = 0$ which passes through the center of the unit hypercube also passes through one of the hypercubes vertices. 
See Figure \ref{Fig1} for a visual reference. 
\begin{figure}
\includegraphics[scale = 0.4]{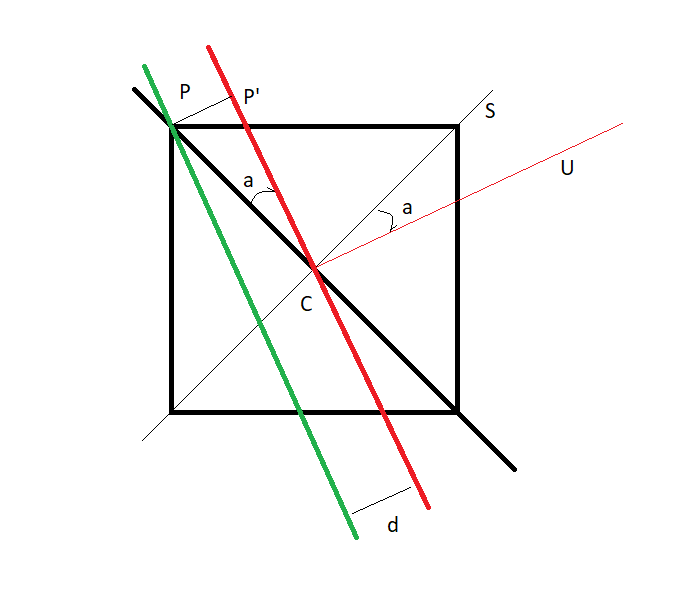}
\caption{Planar representation of approximating one normal vector ($S$ with black) by another vector ($U$, with red). }
\label{Fig1}
\end{figure}

Let $c \in \mathbb{N}$ fixed, $N = n^c$ and $u_k \in \{1, ... , N\}$ defined as below
\begin{align}\label{E2}
\begin{cases}
u_1 = \left\lfloor N \cdot \frac{s_1}{\|S\|} \right\rfloor \\
\vdots \\
u_n = \left\lfloor N \cdot \frac{s_n}{\|S\|} \right\rfloor
\end{cases}
\end{align}  then denote $U = \begin{bmatrix} u_1 &... &u_n\end{bmatrix}^T $. As such $\frac{U}{\|U\|}$ is an approximation of the normal vector $\frac{S}{\|S\|}$  with the advantage that $U$ has elements written on at most $\log_2(N) = c \cdot \log_2(n)$ bits. 
\begin{lemma} \label{L1}

One has the following:
\begin{align}
\lim_{N \to \infty} \frac{N}{\|U\|} = 1
\end{align}
\end{lemma}

\begin{proof}
For each $k \in \{1, \hdots, n\}$ it is obtained
\begin{align}
0 \leq N \cdot \frac{s_k}{\|S\|}  - u_k\leq 1
\end{align} hence

\begin{align}
&0 \leq \frac{s_k}{\|S\|} - \frac{u_k}{N} \leq \frac{1}{N} \nonumber \\
&0 \leq \sum_{k=1}^n \left( \frac{s_k}{\|S\|} - \frac{u_k}{N} \right)^2 \leq \frac{n}{N^2}
\end{align} therefore 
\begin{align}
\left\| \frac{S}{\|S\|} - \frac{U}{N}\right\|^2 &= \left( \frac{S}{\|S\|} - \frac{U}{N} \right)^T \cdot \left( \frac{S}{\|S\|} - \frac{U}{N} \right) \nonumber \\
& = \frac{\|U\|^2}{N^2} + 1 - 2 \cdot \frac{S^T\cdot U}{\|S\|\cdot N} \leq \frac{n}{N^2}.
\end{align} Multiplying both sides with $\frac{N}{\|U\|}$ one gets
\begin{align}\label{E7}
\left(1 - \frac{n}{N^2}\right) \cdot  \frac{N}{\|U\|} + \frac{\|U\|}{N} \leq 2 \cdot \frac{S^T\cdot U}{\|S\|\cdot \|U\|}  \leq 2
\end{align} Letting $x = \frac{N}{\|U\|}$ in (\ref{E7}) one has for any $\epsilon = \frac{n}{N^2}> 0$
\begin{align}
0 \leq (1-\epsilon)\cdot x + \frac{1}{x} \leq 2
\end{align} therefore assuming that exists $\lim_{\epsilon \to 0} x = x_{\infty}$, this meets 
\begin{align}
\frac{x_{\infty}^2-2\cdot x_{\infty} + 1}{x_{\infty}} = 0 \hspace{0.5cm} \Rightarrow \hspace{0.5cm} x_{\infty} = 1.
\end{align} 
\end{proof}

We give a main result below:
\begin{theorem} \label{T1}
For $\frac{n}{N^2}$ small enough, if the partition problem (\ref{E1}) has at least a solution, then the following problem also has at least one

\begin{align}\label{E10}
\exists^?x\in \{0,1\}^n \hspace{0.5cm} \text{with} \hspace{0.5cm} U^T\cdot x = \frac{U^T\cdot 1_{n\times 1}}{2} + \delta 
\end{align} for $\delta \in \mathbb{Z}$ with $-\frac{n}{2} \leq \delta \leq \frac{n}{2}$. 

\end{theorem} 

The above theorem can be used as follows: if there is no solution to (\ref{E10}) then (\ref{E1}) does not have a solution either. 
\begin{proof}
For a solution $P$ of the original partition problem, (i.e. a vertex of the unit hypercube on the hyperplane $S^T\cdot \left( x - \frac{1}{2} \cdot 1_{n \times 1}\right) = 0$ ) let $P'$ be the projection on the hyperplane $U^T\cdot \left( x - \frac{1}{2} \cdot 1_{n \times 1}\right) = 0$ (with red in Figure \ref{Fig1}), also a central hyperplane passing through the cube center, but having $U$ as the normal vector. Define $c_a, d^{\star}$ as follows
\begin{align}\label{E11}
c_a = \frac{\sum_{k=1}^n u_k\cdot s_k}{\|U\| \cdot \|S\|} = \cos(a) \hspace{0.5cm} d^{\star} = \frac{\sqrt{n}}{2} \cdot \sqrt{1 - c_a^2}
\end{align} where $a$ is the angle between $U, S$.

 It follows that $P$ (the original solution) should belong to a translation (with green in Figure (\ref{Fig1}), either forward or backwards of the hyperplane $U^T\cdot \left( x - \frac{1}{2} \cdot 1_{n \times 1}\right) = 0$ ) along its normal direction, with some distance $d \leq d^{\star}$. Therefore for each solution $P$ of the original problem, it must also be the solution of the following problem:
\begin{align}\label{E12}
\exists^?x\in \{0,1\}^n \hspace{0.5cm} \text{with} \hspace{0.5cm} U^T\cdot \left( x - \left( \frac{1}{2}\cdot 1_{n \times 1} + \frac{U}{\|U\|} \cdot d \right)\right) = 0
\end{align} for $-d^{\star} \leq d \leq d^{\star}$

 The later problem is equivalent to 
\begin{align}\label{E13}
\exists^?x\in \{0,1\}^n \hspace{0.5cm} \text{with} \hspace{0.5cm} U^T\cdot x = T 
\end{align} where $T = \frac{\sum_{k=1}^n u_k}{2} \pm d \cdot \|U\| $. All is left to prove is that $d^{\star} \cdot \|U\| \leq \frac{n}{2}$. Indeed, from (\ref{E11}) and Lemma \ref{L1} one has that 
\begin{align}\label{E14}
d^{\star} \cdot \|U\| &= \frac{\sqrt{n}}{2} \cdot \sqrt{1 - c_a^2} = \frac{\sqrt{n} \cdot \|U\|}{2} \cdot \sqrt{ 1 - \left( \frac{\left( 1 - \frac{n}{N^2}\right) \frac{N}{\|U\|} + \frac{\|U\|}{N}}{2}\right)^2} \nonumber \\
& \approx^{\frac{N}{\|U\|} \sim 1} \frac{\sqrt{n}\cdot \|U\|}{2} \cdot \sqrt{1 - \left( 1 - \frac{n}{2\cdot N^2}\right)^2} \nonumber \\
& = \frac{\sqrt{n}\cdot \|U\|}{2} \cdot \sqrt{\frac{n}{2\cdot N^2} \cdot \left( 2 -  \frac{n}{2\cdot N^2} \right)} \nonumber \\
& \leq \frac{\sqrt{n} \cdot \|U\|}{2} \cdot \frac{\sqrt{n}}{N} \approx \frac{n}{2} 
\end{align}
\end{proof}

In general, the converse of Theorem \ref{T1} need not be true, i.e. if the problem (\ref{E12}, \ref{E13}) has a solution, it does not necessarily mean that the original partition problem (\ref{E1}) also has one. We shall show, that anyway, a solution to (\ref{E12}, \ref{E13}) is an approximate solution to (\ref{E1}). 

\begin{theorem}\label{T2} Let $Q \in \{0,1\}^n$ be a solution to (\ref{E12}, \ref{E13}), then $Q$ is an approximate solution to (\ref{E1}) such that
\begin{align}
\left| \frac{S^T \cdot Q^T}{\frac{S^T\cdot 1_{n \times 1}}{2}} - 1 \right| \leq 2\cdot \frac{n}{N}
\end{align} where recall $N = n^c$ for a fixed chosen $c \in \mathbb{N}$. 
\end{theorem}
\begin{proof}
Let $Q$ be a solution to 
\begin{align}
U^T\cdot \left( x - \left( \frac{1}{2}\cdot 1_{n \times 1} + \frac{U}{\|U\|} \cdot d \right)\right) = 0
\end{align} then $Q \in \{0,1\}^n$ belongs to a hyperplane with normal vector $U$ and passing through the point $C = \frac{1}{2}\cdot 1_{n \times 1} + \frac{U}{\|U\|} \cdot d $. Let $Q'$ be its projection on the hyperplane $S^T \cdot \left( x - C \right) = 0$. Then since $\|Q - C\| \leq \frac{\sqrt{n}}{2}$ follows that $\|Q - Q'\| \leq \frac{\sqrt{n}}{2} \cdot \sin(a) = \delta^{\star}$. As such, it must exist $\delta \in [-d^{\star}, d^{\star}]$ such that 
\begin{align}
S^T \cdot \left( Q - \left( C + \delta \cdot \frac{S}{\|S\|}\right)\right) = 0
\end{align} that is
\begin{align}\label{E18}
& S^T \cdot Q = \frac{S^T \cdot 1_{n \times 1}}{2} + \frac{S^T\cdot U}{\|U\|} \cdot d + \delta \cdot \|S\| \nonumber \\
& \left| S^T \cdot Q- \frac{S^T \cdot 1_{n \times 1}}{2} \right| \leq  (d+ \delta) \cdot \|S\|
\end{align} Finally 
\begin{align}
\left|\frac{\frac{S^T \cdot 1_{n \times 1}}{2} + (d + \delta ) \cdot \|S\|}{\frac{S^T\cdot 1_{n \times 1}}{2}} -1 \right|&= 2 \cdot \frac{\|S\|}{S^T \cdot 1_{n \times 1}} \cdot (d + \delta) \nonumber \\
& \leq 4 \cdot d^{\star} \leq^{\text{see (\ref{E14})}} \frac{2\cdot n}{N}  
\end{align}
\end{proof}

Based on the above results, we give our main theorem:
\begin{theorem}\label{T3}
For $d^{\star}$ given by (\ref{E11}), there exists an algorithm of complexity $\mathcal{O}\left( N\cdot n^{\frac{5}{2}}\right)$ which upon completion fulfils exactly one of the following alternatives:
\begin{enumerate}
\item proves that the slab $\mathcal{S}\left( S, \frac{1}{2}\cdot 1_{n \times 1},2\cdot d^{\star}   \approx  \frac{n}{N} \right)$ contains no vertex of the hypercube $Q_n = [0,1]^n$. 
\item finds a vertex of the hypercube  $Q_n = [0,1]^n$ in the slab \\$\mathcal{S}\left( S, \frac{1}{2}\cdot 1_{n \times 1}, 8\cdot d^{\star}   \approx  \frac{4\cdot n}{N} \right)$. 
\end{enumerate}
\end{theorem}

\begin{proof}
The proposed algorithm is solving the following problem 
\begin{align}\label{E21a}
\exists^?x\in \{0,1\}^n \hspace{0.5cm} \text{with} \hspace{0.5cm} U^T\cdot x = \frac{U^T\cdot 1_{n\times 1}}{2} + t \in \mathbb{N}
\end{align} for $t \in \mathbb{Z}$ with $-n \leq t \leq n$ using Dynamical Programming. The stated complexity shall be proven later. 

 First assume that there is a point in the slab $\mathcal{S}\left( S, \frac{1}{2}\cdot 1_{n \times 1},2\cdot d^{\star}   \approx  \frac{n}{N} \right)$ and we show that it can be found as the solution to one of the problems in (\ref{E21a}).

 Let $P \in \mathcal{S}\left(S, \frac{1}{2} \cdot 1_{n \times 1}, 2\cdot d^{\star}\right)$ then 
\begin{align}
-d^{\star} \leq \frac{S^T}{\|S\|} \cdot \left(P - \frac{1}{2} \cdot 1_{n \times 1} \right) \leq d^{\star} 
\end{align}  therefore it exists $d \in \left[-d^{\star}, d^{\star} \right]$ such that 
\begin{align}
\frac{S^T}{\|S\|} \cdot \left(P - \left( \frac{1}{2} \cdot 1_{n \times 1} + \frac{S}{\|S\|} \cdot d \right) \right) = 0
\end{align} Let $U$ be given by (\ref{E2}) and $P'$ be the projection of $P$ on the hyperplane $\frac{U^T}{\|U\|} \cdot \left(x -\left( \frac{1}{2} \cdot 1_{n \times 1} + \frac{S}{\|S\|}\cdot d\right) \right) = 0$. Denote
\begin{align}
C = \frac{1}{2} \cdot 1_{n \times 1} + \frac{S}{\|S\|}\cdot \delta
\end{align} then $\|C - P\| \leq \frac{\sqrt{n}}{2}$ and $\|P - P'\| \leq \sin(\angle(U,S)) \cdot \frac{\sqrt{n}}{2} = d^{\star}$ from (\ref{E11}). 

It follows that $P$ belongs to a parallel translation of the hyperplane $\frac{U^T}{\|U\|} \cdot \left(x -\left( \frac{1}{2} \cdot 1_{n \times 1} + \frac{S}{\|S\|}\cdot d\right) \right) = 0$, hence it exists $\delta \in [-d^{\star}, d^{\star}]$ such that 
\begin{align}
\frac{U^T}{\|U\|} \cdot \left(P -\left( \frac{1}{2} \cdot 1_{n \times 1} + \frac{S}{\|S\|}\cdot d + \frac{U}{\|U\|} \cdot \delta \right) \right) = 0
\end{align} that is
\begin{align}
U^T\cdot P = \frac{U^T\cdot 1_{n\times 1}}{2} + \frac{U^T\cdot S}{\|S\|} \cdot d + \|U\| \cdot \delta.  
\end{align} Because
\begin{align}
\left| \frac{U^T\cdot S}{\|S\|} \cdot d + \|U\| \cdot \delta \right| \leq \left( d + \delta \right) \cdot \|U\| \leq 2 \cdot d^{\star} \cdot \|U\| \approx^{(\ref{E14})} n 
\end{align} the conclusion follows.

Next, for the second alternative, we show that if the problem

\begin{align}\label{E28}
\exists^?x\in \{0,1\}^n \hspace{0.5cm} \text{with} \hspace{0.5cm} U^T\cdot x = \frac{U^T\cdot 1_{n\times 1}}{2} + t \in \mathbb{N}
\end{align} for $t \in \mathbb{Z}$ with $-n \leq t \leq n$ has a solution $Q$ then from Theorem \ref{T2} one can readily see that 
 $Q \in \mathcal{S}\left(S, \frac{1}{2}\cdot 1_{n \times 1}, 4 \cdot d^{\star} \approx \frac{n}{N} \right)$. Indeed, from (\ref{E18}) 
\begin{align}
& S^T \cdot Q = \frac{S^T \cdot 1_{n \times 1}}{2} + d \cdot \|S\| + \delta \cdot \frac{S^T\cdot U}{\|U\|} \nonumber \\
& \left| S^T\cdot Q - \frac{S^T \cdot 1_{n \times 1}}{2} \right| \leq (\delta + d) \cdot \|S\| \leq 4 \cdot d^{\star} \cdot \|S\|
\end{align} that is
\begin{align}
\frac{S}{\|S\|} \cdot \left( Q -  \frac{1}{2}\cdot 1_{n \times 1} \right) = \rho \in[-n, n] = \left[-4\cdot d^{\star}, 4\cdot d^{\star}\right]
\end{align} hence, as stated $Q \in \mathcal{S}\left(S, \frac{1}{2}\cdot 1_{n \times 1}, 8 \cdot d^{\star} \approx \frac{4\cdot n}{N} \right) $. 

\end{proof}
\section{Complexity analysis}

The proof of Theorem \ref{T3} is completed by the following remark regarding the complexity of solving the problems (\ref{E12}, \ref{E13})
\begin{remark}
Solving any instance of the subset-sum in the following form:
\begin{align}
U^T\cdot x = \frac{U^T\cdot 1_{n \times 1}}{2} + t
\end{align} with $t \in [-n, n]$ can be done in $\mathcal{O}\left( n \cdot T+t\right) $ using dynamical programming. Since $T + t \leq \frac{U^T \cdot 1_{n \times 1}}{2} + n \leq \frac{\|U\| \cdot \sqrt{n}}{2} + n \in \mathcal{O}\left( N \cdot \sqrt{n} \right)$ follows that solving one instance (of the $2\cdot n$ instances) has complexity
\begin{align}
\mathcal{O} \left( N \cdot n^{\frac{3}{2}}\right)
\end{align} hence solving them all requires complexity $\mathcal{O}\left( N \cdot n^{\frac{5}{2}} \right)$
\end{remark} 

\section{Application to Simultaneous Subset-sum Problem}

Let $p \in \mathbb{N}$ and consider $S_i \in \mathbb{N}^n$ and $T_i \in \mathbb{N}$ for each $i \in \{1, \hdots, p\}$.  Consider the Simultaneous Subset Sum Problem (SSSP):
\begin{align}\label{E37}
\exists^? x \in \{0,1\}^n \hspace{0.5cm} \text{with} \hspace{0.5cm} \begin{cases} S_1^T\cdot x = T_1 \\ 
\vdots \\
S_p^T \cdot x = T_p
\end{cases}
\end{align} 

Let $C = \frac{1}{2} \cdot 1_{n \times 1}$ then the closed ball $\bar{\mathcal{B}} \left(C, \frac{\sqrt{n}}{2} \right)$ contains the vertices of $\mathcal{Q}_n = [0,1]^n$.  We assume in the following that $T_i = \frac{S_i^T\cdot 1_{n \times 1}}{2}$, i.e. what is analyzed in the following is can be referred to as the Simultaneous Subset Partition Problem (SSPP).

\subsection{Preliminary Results}
The problem (\ref{E37}) can be relaxed as below:
\begin{align}\label{E38}
\exists^? x \in \{0,1\}^n \hspace{0.5cm} \text{with} \hspace{0.5cm} \begin{cases}  -\frac{\epsilon}{2} \leq \frac{S_1^T}{\|S_1\|}\cdot \left( x  - C \right) \leq \frac{\epsilon}{2} \\ 
\vdots \\
-\frac{\epsilon}{2} \leq \frac{S_p^T}{\|S_p\|}\cdot \left( x  - C \right) \leq \frac{\epsilon}{2}
\end{cases}
\end{align} Let $\rho > 0$ fixed and $R =\sqrt{ \rho^2 + \frac{n}{4}}$ then define:
\begin{align}
C_i = C - \rho \cdot \frac{S_i}{\|S_i\|}
\end{align} and define the problem:

\begin{align}\label{E40}
\exists^? x \in \{0,1\}^n \hspace{0.5cm} \text{with} \hspace{0.5cm} \begin{cases}  -\frac{\delta}{2} \leq \|x - C_1\| - R \leq \frac{\delta}{2} \\ 
\vdots \\
-\frac{\delta}{2} \leq \|x - C_p\| - R \leq \frac{\delta}{2}
\end{cases}
\end{align}

The following lemma justifies the approach that shall be given in the following
\begin{lemma}\label{L4.1}
Let $x$ be a solution to (\ref{E38}) then $x$ is also a solution of (\ref{E40}) with $\delta \leq 2\cdot \left( \frac{\epsilon}{2} + \frac{n}{8\cdot \rho}\right)$. Conversely, if $x$ is a solution to (\ref{E40}) then $x$ is also a solution to (\ref{E38}) with $\epsilon \leq 2 \cdot \left( \frac{\delta}{2} + \frac{n}{8 \cdot \rho}\right)$

\end{lemma}
\begin{proof}
One can see in Figure \ref{Fig2} that

\begin{figure}
\includegraphics[scale = 0.6]{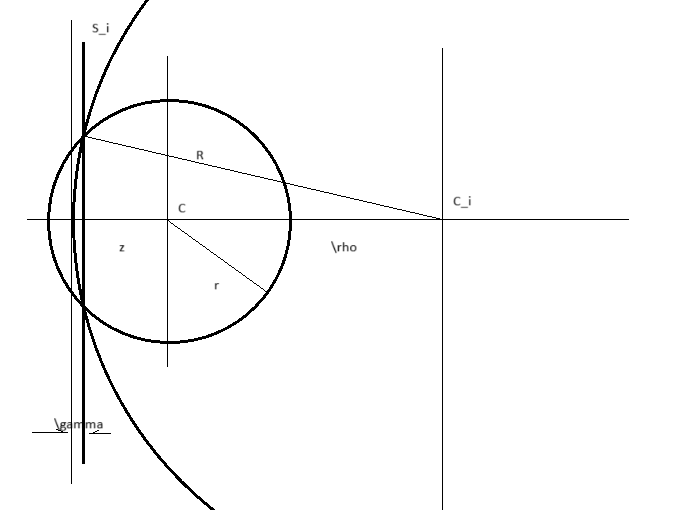}
\caption{Planar representation of local plane approximation by sphere}
\label{Fig2}
\end{figure}

\begin{align}
\gamma = R - \left( z + \rho \right) 
\end{align} In the lemma hypothesis case, $z = 0$ since the hyperplanes pass through the cube center $C$ it is obtained
\begin{align}\label{E42}
\gamma = \sqrt{ \rho^2 + \frac{n}{4} } - \rho = \frac{\frac{n}{4}}{ \sqrt{\rho^2 + \frac{n}{4}} + \rho}  \leq \frac{n}{8 \cdot \rho}
\end{align} 

It follows that all $x$ in the ball $\mathcal{B}\left(C,\frac{\sqrt{n}}{2} \right) $ that are close to the hyperplane $S_i^T \cdot (x - C) = 0$ are also close to the boundary of the ball $\{x |\| x - C\| \leq R^2\}$. 

Let $x$ be a solution to (\ref{E40}), then exists $y \in \mathcal{B}\left(x, \frac{\epsilon}{2} \right)$ such that $S_i^T\cdot (y - C) = 0$, hence (from (\ref{E42})) $y \in \mathcal{B}(C_i, R\pm \gamma)$ for all $i \in \{1, \hdots, p\}$. Therefore $x \in \mathcal{B}(C_i, R\pm \gamma \pm \frac{\epsilon}{2})$ for all $i \in \{1, \hdots, p\}$, that is $x$ is a solution to (\ref{E40}) with $\delta$ as claimed.

Similarly let $x$ be a solution to (\ref{E40}). Then exists $x \in \mathcal{B}\left(C_i, R \pm \frac{\delta}{2} \right)$ hence, from (\ref{E42}) follows that for each $i \in \{1, \hdots, p\} $ $x$ belongs to the slab $\mathcal{S}\left(S_i,C,2 \cdot \left( \frac{\delta}{2}+ \frac{n}{8 \cdot \rho}\right) \right)$ i.e. $x$ is a solution to (\ref{E38}) with $\epsilon$ as claimed. 

\end{proof}

From (\ref{E40}) one obtains the equivalent problem:

\begin{align}\label{E43}
\exists^? x \in \{0,1\}^n \hspace{0.5cm} \text{with} \hspace{0.5cm} \begin{cases}  -\frac{\delta}{2}\cdot (\|x - C_1\| + R) \leq \|x - C_1\|^2 - R^2 \leq \frac{\delta}{2}\cdot (\|x - C_1\| + R) \\ 
\vdots \\
-\frac{\delta}{2}\cdot (\|x - C_p\| + R) \leq \|x - C_p\|^2 - R^2 \leq \frac{\delta}{2}\cdot (\|x - C_1\| + R)
\end{cases}
\end{align} hence we can finally define the problem that shall be studied in the following:

\begin{align}\label{E43}
\exists^? x \in \{0,1\}^n \hspace{0.5cm} \text{with} \hspace{0.5cm} \begin{cases}  -\frac{\delta}{2} \leq \|x - C_1\|^2 - R_1^2 \leq \frac{\delta}{2}\\ 
\vdots \\
-\frac{\delta}{2}\leq \|x - C_p\|^2 - R_p^2 \leq \frac{\delta}{2}
\end{cases}
\end{align} 

\subsection{The Fundamental Result for Simultaneous Subset Sum}
Let us define:
\begin{align}\label{E45}
L_0(x) = \sum_{i = 1}^p \left( \|x - C_i\|^2 - R_i^2 \right)^2
\end{align} and we have the main theorem of this section:

\begin{theorem} Assume that $p < n$ is a power of two. If $\left\| \frac{\sum_{i=1}^p C_i}{p} - C \right\| $ is large enough, then for a given $\delta > 0$ there exists an algorithm of quasi-polynomial complexity which finds  $ x \in \{0,1\}^n$ such that $L_0(x) \in \mathcal{O}( \delta)$ if \textit{exactly one} exists. Failing to return a solution therefore means that either no solution exists either there are at least two. 
\end{theorem}

\begin{proof}
First note that
\begin{align}\label{E46}
\left( \|x - C_1\|^2 - R_1^2 \right)^2 + \left( \|x - C_2\|^2 - R_2^2 \right)^2 &= \left( \|x - C_1\|^2 - R_1^2 + \|x - C_2\|^2 - R_2^2  \right)^2 - \nonumber \\
& - 2 \cdot \left( \|x - C_1\|^2 - R_1^2 \right) \cdot \left( \|x - C_2\|^2 - R_2^2  \right)
\end{align} then

\begin{align}\label{E47}
 &\|x - C_1\|^2 - R_1^2 + \|x - C_2\|^2 - R_2^2  = \|x\|^2 + \|C_1\|^2 - 2 \cdot x^T \cdot C_1 + \|x\|^2 + \|C_2\|^2 - 2 \cdot x^T \cdot C_2 - R_1^2 - R_2^2\nonumber \\
& = 2 \cdot \left( \| x \|^2 - 2 \cdot x^T \cdot \frac{C_1 + C_2}{2} + \frac{\|C_1 + C_2\|^2}{4} \right) - \frac{\|C_1 + C_2\|^2}{2} + \|C_1\|^2 + \|C_2\|^2 - R_1^2 - R_2^2 \nonumber \\
& = 2 \cdot \left\| x - \frac{C_1 + C_2}{2} \right\|^2 + \frac{1}{2} \cdot \left\| C_1 - C_2\right\|^2 - R_1^2 - R_2^2 \nonumber \\
& = 2 \cdot \left( \left\| x-  C_{1,2}\right\|^2 - R_{1,2}^2\right)
\end{align} with $C_{1,2} =  \frac{C_1 + C_2}{2}$ and $R_{1,2}^2= \frac{R_1^2 + R_2^2}{2} - \frac{\|C_1-C_2\|^2}{4}$ and 

Replacing (\ref{E47}) in (\ref{E46}) one gets:
\begin{align}\label{E48}
\left( \|x - C_1\|^2 - R_1^2 \right)^2 + \left( \|x - C_2\|^2 - R_2^2 \right)^2 &= 4 \cdot  \left( \left\| x - C_{1,2}\right\|^2 - R_{1,2}^2 \right)^2 - 2 \cdot \left( \|x - C_1\|^2 - R_1^2\right) \cdot \left( \|x - C_2\|^2 - R_2^2 \right) \nonumber \\
& = 4 \cdot  \left( \left\| x - C_{1,2}\right\|^2 - R_{1,2}^2 \right)^2 - M_{1,2}(x)
\end{align}

where 
\begin{align}\label{E49}
M_{1,2}(x) =  2 \cdot \left( \|x - C_1\|^2 - R_1^2\right) \cdot \left( \|x - C_2\|^2 - R_2^2 \right) 
\end{align} 

As such the sum $\left(\|x - C_1\|^2 - R_1^2 \right)^2 + \left(\|x - C_2\|^2 - R_2^2 \right)^2$ is replaced in (\ref{E48}) by a formula involving only one ball.

Replacing all pairs in $L_0(x)$ with an expression similar to (\ref{E48}) one gets 
\begin{align}\label{E50}
L_0(x) = 4 \cdot\sum_{i}  \left( \left\| x - C_{i,i+1}\right\|^2 - R_{i,i+1}^2 \right)^2 - \sum_{i} M_{i,i+1}(x) 
\end{align}

Since a solution would be $x \in \{0,1\}^n$ then from (\ref{E49})
\begin{align}\label{E51}
\left| M_{1,2}(x)\right| = 2 \cdot \left| \|x - C_1\|^2 - R_1^2\right| \cdot \left| \|x - C_2\|^2 - R_2^2 \right|  \leq \bar{M}_{1,2} \in \mathcal{O} \left( \text{poly}(n, \rho) \right)
\end{align} therefore 
\begin{align}
\left| \sum_{i} M_{i,i+1}(x) \right| \leq \bar{M}_0 \in \mathcal{O} (\text{poly}(n, \rho))
\end{align} From (\ref{E55}) one obtains a possible approximation of $L_0(x)$ 

\begin{align}
\hat{L}_0(x) = 4 \cdot\sum_{i}  \left( \left\| x - C_{i,i+1}\right\|^2 - R_{i,i+1}^2 \right)^2  - M_0
\end{align} where $M_0 \in \left\{-\bar{M}_{0}, -\bar{M}_{0} + \frac{\delta}{\log_2(p)}, \hdots, \bar{M}_{0} \right\} $ i.e. $\mathcal{O} \left( \frac{\log_2(p)\cdot \text{poly}(n,\rho)}{\delta} \right)$ possible values. The reasoning is that if $x$ is a solution, then exists $M_0$ such that $\left| \sum_{i} M_{i,i+1}(x) - M_0 \right| \leq \frac{\epsilon}{\log_2(p)}$. 

Define
\begin{align}
L_1(x) = \sum_{i}  \left( \left\| x - C_{i,i+1}\right\|^2 - R_{i,i+1}^2 \right)^2
\end{align} and note that it has only half the terms of $L_0(x)$ from (\ref{E45}). The whole process is reiterated to obtain

\begin{align}\label{E55}
L_1(x) = 4 \cdot\sum_{i}  \left( \left\| x - C_{i,i+1,i+2,i+3}\right\|^2 - R_{i,i+1, i+2, i+3}^2 \right)^2 - \sum_{i} M_{i,i+1, i+2, i+3}(x). 
\end{align} Approximate $L_1(x)$ by 
\begin{align}
\hat{L}_1(x) = 4 \cdot\sum_{i}  \left( \left\| x - C_{i,i+1,i+2,i+3}\right\|^2 - R_{i,i+1,i+2,i+3}^2 \right)^2  - M_1
\end{align} where $M_1 \in \left\{-\bar{M}_{1}, -\bar{M}_{1} + \frac{\delta}{4^1 \cdot \log_2(p)}, \hdots, \bar{M}_{1} \right\} $, $\bar{M}_1$ defined similar to $\bar{M}_0$ as an upper bound on $\left| \sum_i M_{i,i+1,i+2,i+3}(x)\right| $ for $x \in \{0,1\}^n$  and $ C_{i,i+1,i+2,i+3} = \frac{C_{i,i+1} + C_{i+2,i+3}}{2}$. 

One can make at most $\log_2(p)$ such steps until it reaches
\begin{align}
L_{\log_2(p)}(x) = 4 \cdot \left( \left\| x - C_{1,2, \hdots, p}\right\|^2 - R_{1, 2, \hdots, p}^2 \right)^2 - M_{1, 2, \hdots, p}(x) 
\end{align} which is approximated by 

\begin{align}
\hat{L}_{\log_2(p)}(x) = 4 \cdot \left( \left\| x - C_{1,2, \hdots, p}\right\|^2 - R_{1, 2, \hdots, p}^2 \right)^2 - M_{\log_2(p)} 
\end{align} where  $M_{\log_2(p)} \in \left\{-\bar{M}_{\log_2(p)}, -\bar{M}_{\log_2(p)} + \frac{\delta}{4^{\log_2(p)} \cdot \log(n)}, \hdots, \bar{M}_{\log_2(p)} \right\} $, $\bar{M}_{\log_2(p)}$ is defined as an upper bound on $\left|M_{1, 2, \hdots, p}(x) \right| $ for $x \in \{0,1\}^n$  and $ C_{1, \hdots, p} = \frac{C_{1, \hdots, \frac{p}{2}} + C_{\frac{p}{2}+1, \hdots, p}}{2}$. Finally one gets

\begin{align}\label{E59}
L_0(x) &\approx 4 \cdot L_1(x) - M_0 \nonumber \\
& \approx 4 \cdot \left( 4 \cdot L_2(x) - M_1 \right) - M_0 \nonumber \\ 
& \approx \vdots \nonumber \\
& \approx 4^{\log_2(p)} L_{\log_2(p)}(x) - \sum_{q = 0}^{\log_2(p)-1} 4^q \cdot M_q  \nonumber \\
& \approx 4^{\log_2(p)+1} \cdot \left( \left\| x - C_{1,2, \hdots, p}\right\|^2 - R_{1, 2, \hdots, p}^2 \right)^2 - \sum_{q = 0}^{\log_2(p)} 4^q \cdot M_q  =: \tilde{L}_0(x,M_0, \hdots, M_{\log_2(p)})
\end{align}

Note that if $x \in \{0,1\}^n$ then exists $M_0, \hdots, M_{\log_2(p)}$ such that 
\begin{align}\label{E60}
\left| L_0(x) - \tilde{L}(x)\right| \leq \sum_{i=0}^{\log_2(p)} 4^i \cdot \frac{\delta}{4^i \cdot \log_2(p)} = \delta \cdot \frac{\log_2(p)+1}{\log_2(p)} \leq 2 \cdot \delta 
\end{align} Recall that $M_q \in \left\{-\bar{M}_q, -\bar{M}_q+\frac{\delta}{4^q \cdot \log_2(p)}, \hdots, \bar{M}_q \right\}$ therefore, because $\bar{M}_q \in \mathcal{O}(\text{poly}(n, \rho))$, $M_q$ can take $\mathcal{O} \left( \frac{\text{poly}(n, \rho) \cdot 4^i \cdot \log_2(p)}{\delta}\right)$ values. This is upper bounded by $\mathcal{O} \left( \frac{\text{poly}(n, \rho)}{\delta}\right)$ values. 

Consider the following problem:
\begin{align}\label{E61}
\exists^? x \in \{0,1\}^n, M_q \in  \left\{-\bar{M}_q, -\bar{M}_q+\frac{\delta}{4^q \cdot \log_2(p)}, \hdots, \bar{M}_q \right\} \hspace{0.5cm} \text{with} \hspace{0.5cm} \tilde{L}(x, M_0, \hdots, M_{\log_2(p)}) \leq 3\cdot \delta
\end{align} From (\ref{E60}) follows that if exists $x \in \{0,1\}^n $ with $L_0(x) \leq \delta$ then exists $M_q$ such that $\tilde{L}_0(x, \hdots) \leq 3 \cdot \delta$. On the other hand, if exists $M_q's$ such that for some $x \in \{0,1\}^n $ one has $\tilde{L}_0(x) \leq 3\cdot \delta$ then $L_0(x) \leq 5 \cdot \delta$. 

We end the proof with the quasi-polynomial algorithm for solving (\ref{E61}). In order to do so, one takes all values of $M_q$ and solves $\tilde{L}_0(x, \hdots) \leq 3\cdot \delta$ as follows. There will be at most $\mathcal{O}\left( \left(\frac{\text{poly}(n, \rho)}{\delta} \right)^{\log_2(p)+1}\right)$ such problems. For fixed $M_0, \hdots, M_{\log_2(p)}$ one has $\tilde{L}_0(x, \hdots) \leq \delta$ is, see (\ref{E59})
\begin{align}
\sqrt{\frac{-\delta + \sum_{q=1}^{\log_2(p)+1} 4^q \cdot M_q }{4^{\log_2(p)+1}}} \leq \|x - C_{1, \hdots, p}\|^2 - R_{1, \hdots, p}^2 \leq \sqrt{\frac{\delta + \sum_{q=1}^{\log_2(p)+1} 4^i \cdot M_i }{4^{\log_2(p)+1}}} 
\end{align} which is 
\begin{align}\label{E63}
\frac{\sqrt{\frac{-\delta + \sum_{q=1}^{\log_2(p)+1} 4^q \cdot M_q }{4^{\log_2(p)+1}}} }{\|x - C_{1, \hdots, p}\| + R_{1, \hdots, p}} \leq \|x - C_{1, \hdots, p}\| - R_{1, \hdots, p} \leq \frac{\sqrt{\frac{\delta + \sum_{q=1}^{\log_2(p)+1} 4^q \cdot M_q }{4^{\log_2(p)+1}}}}{\|x - C_{1, \hdots, p}\| + R_{1, \hdots, p}}
\end{align} Since 
\begin{align}
B_L = \left| \frac{\sqrt{n}}{2} - \|\frac{1}{2}\cdot 1_{n \times 1} - C_{1, \hdots, p}\| \right| + R_{1, \hdots, p} \leq \|x - C_{1, \hdots, p}\| + R_{1, \hdots, p} \leq \frac{\sqrt{n}}{2} + \|\frac{1}{2}\cdot 1_{n \times 1} - C_{1, \hdots, p}\| + R_{1, \hdots, p}  = B_U
\end{align} then for some $\epsilon > 0$ let $B \in \{B_L, B_L + \epsilon, \hdots, B_U\}$. If exists $x \in \{0,1\}^n$ a solution to (\ref{E63}), then exists $B$ with $|B - \|x - C_{1, \hdots, p}\| - R_{1, \hdots, p}| \leq \epsilon$. To seach for such an $x$ we therefore solve

\begin{align}\label{E64}
\frac{\sqrt{\frac{-\delta + \sum_{q=1}^{\log_2(p)+1} 4^q \cdot M_q }{4^{\log_2(p)+1}}} }{B} \leq \|x - C_{1, \hdots, p}\| - R_{1, \hdots, p} \leq \frac{\sqrt{\frac{\delta + \sum_{q=1}^{\log_2(p)+1} 4^q \cdot M_q }{4^{\log_2(p)+1}}}}{B} \hspace{0.5cm} \forall B \in \{B_L, B_L + \epsilon, \hdots, B_U\}
\end{align} i.e. solving $\frac{B_U-B_L}{\epsilon}$ problems. Each of them is solved using the above results in Lemma \ref{L4.1} and Theorem \ref{T3}. If for any $B^{\star}$ a solution is found, the solution is indeed validated if $|B - \|x - C_{1, \hdots, p}\| - R_{1, \hdots, p}| \leq \epsilon$. 

The requirement in hypothesis about $C_{1, \hdots, p} = \frac{\sum_{i=1}^p C_i}{p}$ being far enough from $Q_n = [0,1]^n$ assures that the approximation of the sphere $\partial \mathcal{B}\left( C_{1, \hdots, p}, R_{1, \hdots, p} \right)$ inside the ball $\mathcal{B}\left( \frac{1}{2}\cdot 1_{n \times 1}, \frac{\sqrt{n}}{2}\right)$ by the hyperplane yields a small error, as presented in Lemma \ref{L4.1}. 

Finally, let $x$ be a solution to (\ref{E63}). Then if $\left| M_q - \sum_{i} M_{i, \hdots, i+2^q-1}(x)\right| \leq \frac{\delta}{4^q \cdot \log_n(p)}$ for all $q \in \{0, \hdots, \log_2(p)\}$ it follows that $L_0(x) \in \mathcal{O}(\delta)$. 

Otherwise, if no solution is found to (\ref{E63}), or none found meets $\left| M_q - \sum_{i} M_{i, \hdots, i+2^q-1}(x)\right| \leq \frac{\delta}{4^q \cdot \log_n(p)}$ for all $q \in \{0, \hdots, \log_2(p)\}$, then one can claim that either there is no $x \in \{0,1\}^n$ for which $L_0(x) \in \mathcal{O}(\delta)$, either there are at least two.  

\end{proof}

\section{Conclusion}
In this paper the Subset Sum Problem (SSP) was analyzed from a geometrical perspective. It was shown that it is equivalent with checking if a certain slice of the positive unit hypercube contains any vertex of the hypercube. 

We showed that the computational difficulty of this problem is due to the fact that the normal vector of the hyperplane associated to the SSP with entries on $m$ bits can have exponentially close vertices which are not solutions. We relaxed the SSP by replacing the infinitely thin slice of the hypercube with a "thicker" one that we've defined and called "slab". 

 Our approach to the relaxed problem was to approximate the normal vector with a vector with entries on less bits. As such we gave a FPTAS which either proves that an "$\epsilon$-thick slice" does not contain any vertex, either finds a vertex in the "$4\cdot \epsilon$-thick slice" of the positive unit hypercube. 

Finally, the approach used so far, i.e. the relaxation of the geometric problem of searching for a vertex of the hypercube in the hyperplane sectioning it, to the problem of searching for a vertex in a slab, is further developed. It is shown that a hyperplane sectioning the hypercube can be approximated with a sphere of large radius centered far enough. As such the slab is then approximated by a thick spherical cap. 

The Simultaneous Subset Sum Problem requires the search for a hypercube vertex in the intersection of say $p$ hyperplanes. This is relaxed to the search for a hypercube vertex in intersection of the associated slabs then the intersection of the associated thick spherical shells. We prove the existence of a quasi-polynomial algorithm for a particular case of the later problem. For $\delta > 0$ arbitrary chosen, the proposed algorithm has complexity $\mathcal{O} \left( \text{poly}\left( \frac{1}{\delta}, n \right) \right)^{\log_2(p)}$ and either gives a point in the intersection of the $p$ $\mathcal{O}(\delta)$-thick spherical shells either, by failing to do so, proves that either no such point exists either that there exists at lest two such points.

\end{document}